\documentclass[final,11pt]{amsart}
\usepackage{amsmath,amssymb,amsfonts} 
\usepackage{epsfig} \usepackage{latexsym,nicefrac,bbm}
\usepackage{xspace}
\usepackage{color,fancybox,graphicx,url,subfigure,fullpage}
\usepackage[top=1in, bottom=1in, left=1in, right=1in]{geometry}
\usepackage{tabularx} \usepackage{hyperref} 
\usepackage{pdfsync}

\usepackage{subfigure}
\usepackage{comment}

\def\showauthornotes{0}

\newcommand{\defeq}{\stackrel{\textup{def}}{=}}
\newcommand{\nfrac}{\nicefrac}
\newcommand{\eps}{\epsilon}

\renewcommand{\epsilon}{\varepsilon}
\newcommand{\Authorcomment}[2]{{\sf \small\color{red}{[#1: #2]}}}

\def\showauthornotes{1} 
 
\def\showdraftbox{1}

\newtheorem{theorem}{Theorem}[section]

\newtheorem{definition}[theorem]{Definition}
\newtheorem{lemma}[theorem]{Lemma}

\def\FullBox{\hbox{\vrule width 6pt height 6pt depth 0pt}}

\def\qed{\ifmmode\qquad\FullBox\else{\unskip\nobreak\hfil
\penalty50\hskip1em\null\nobreak\hfil\FullBox
\parfillskip=0pt\finalhyphendemerits=0\endgraf}\fi}

\def\qedsketch{\ifmmode\Box\else{\unskip\nobreak\hfil
\penalty50\hskip1em\null\nobreak\hfil$\Box$
\parfillskip=0pt\finalhyphendemerits=0\endgraf}\fi}

\newenvironment{proofof}[1]{\begin{trivlist} \item {\bf Proof
#1:~~}}
  {\qed\end{trivlist}}

\newcommand{\marginlabel}[1]%
{\mbox{}\marginpar{\it{\raggedleft\hspace{0pt}#1}}}

\definecolor{Mygray}{gray}{0.8}

 \ifcsname ifcommentflag\endcsname\else
  \expandafter\let\csname ifcommentflag\expandafter\endcsname
                  \csname iffalse\endcsname
\fi

\ifnum\showauthornotes=1
\newcommand{\todo}[1]{\colorbox{Mygray}{\color{red}#1}}
\else
\newcommand{\todo}[1]{}
\fi

\ifnum\showauthornotes=1
\newcommand{\Authornote}[2]{{\sf\small\color{red}{[#1: #2]}}}
\newcommand{\Authoredit}[2]{{\sf\small\color{red}{[#1]}\color{blue}{#2}}}
\newcommand{\Authorfnote}[2]{\footnote{\color{red}{#1: #2}}}
\newcommand{\Authorfixme}[1]{\Authornote{#1}{\textbf{??}}}
\newcommand{\Authormarginmark}[1]{\marginpar{\textcolor{red}{\fbox{
#1:!}}}}
\else
\newcommand{\Authornote}[2]{}
\newcommand{\Authoredit}[2]{}
\newcommand{\Authorcomment}[2]{}
\newcommand{\Authorfnote}[2]{}
\newcommand{\Authorfixme}[1]{}
\newcommand{\Authormarginmark}[1]{}
\fi

\newlength{\pgmtab}  
 {
	\begin{enumerate}}{\end{enumerate}}

\newcounter{lecnum}

\newlength{\tpush}
\ifnum\showdraftbox=1

\else

\fi

\begin{document}
\title{\bf On the Computational Complexity of  Limit Cycles in Dynamical Systems}
\date{}

\begin{abstract}
We study the Poincar\'e-Bendixson theorem for two-dimensional continuous  dynamical systems in compact domains from the point of view of computation, seeking algorithms for finding the limit cycle promised by this classical result.  We start by considering a discrete analogue of this theorem and show that both finding a point on a limit cycle, and determining if a given point is on one, are {\bf PSPACE}-complete. 
 For the continuous version,  we show that both problems are uncomputable in the real complexity sense; i.e., their complexity is  arbitrarily high.  Subsequently, we introduce a  notion of an {\em approximate cycle} and  prove an {\em approximate Poincar\'e-Bendixson theorem} guaranteeing that some orbits come very close to forming a cycle in the absence of approximate fixpoints; surprisingly, it holds for all dimensions.  The corresponding computational problem defined in terms of arithmetic circuits is {\bf PSPACE}-complete.   
 \end{abstract}

\author{Christos H. Papadimitriou}
\thanks{University of California, Berkeley} 

\author{Nisheeth K. Vishnoi}
\thanks{\'{E}cole Polytechnique F\'{e}d\'{e}rale de Lausanne (EPFL)}

\maketitle

\vspace{-12mm}

\begin{figure}[h]
\hspace{1.6cm}
{\includegraphics[height=5cm]{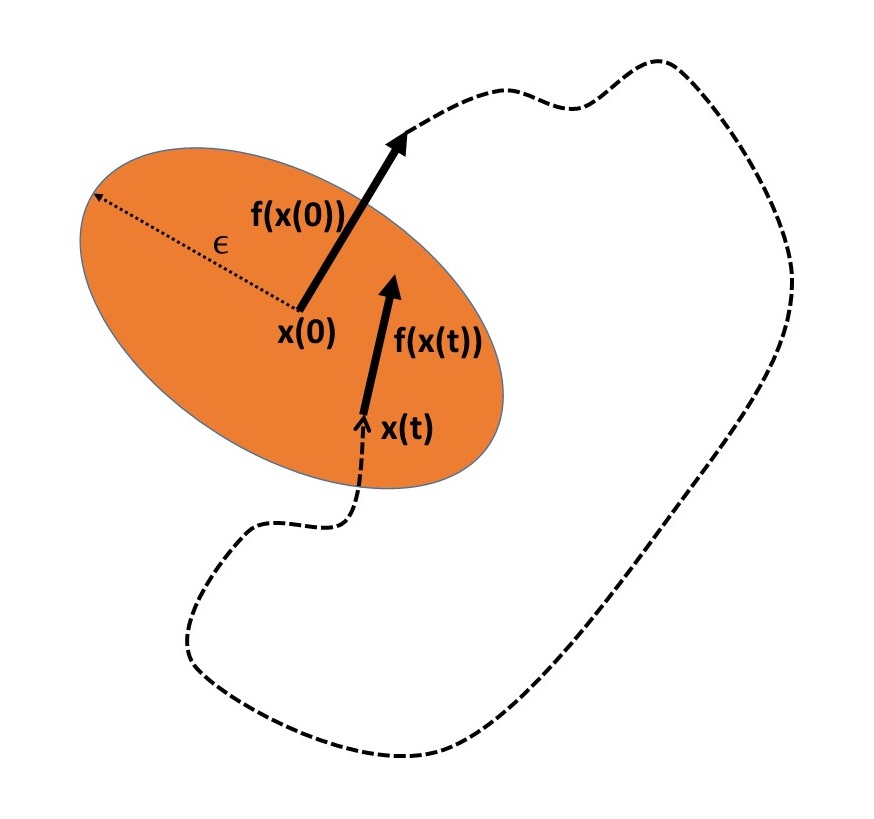}} 
\end{figure}

\vspace{-10mm}

\tableofcontents

\thispagestyle{empty}

\newpage

\setcounter{page}{1}
\section{Introduction}
Dynamical systems are ubiquitous in the study of physical, biological, and social phenomena.  A continuous time dynamical system describes the evolution of a process $x$ through a differential equation $\dot{x} = f(x),$
where $f:\mathbb{R}^n \rightarrow \mathbb{R}^n$ and $\dot{x} \defeq \left( \frac{{\rm d}x_1}{{\rm d}t}, \ldots, \frac{{\rm d}x_n}{{\rm d}t}\right)$, where $n$ is the dimension.
Linear dynamical systems such as $f(x)=Ax+b$ can be completely described by the eigenvalues and eigenvectors of the matrix $A.$  However, linear systems fail to capture many important phenomena, and hence the theory of dynamical systems is primarily concerned with nonlinear functions, typically assumed to be analytically nice; continuous, differentiable, etc.
We are particularly interested in dynamical systems in which the domain of $f$ is a compact subset of $\mathbb{R}^n$, such as a simplex. 

A dynamical system gives rise to a set of {\em trajectories};  a set of values of $x(t)$ for $t>0$ for a given $x(0)$. Under appropriate assumptions, such trajectories are unique given $x(0)$.
Understanding a dynamical system entails understanding the limiting behavior  
of its trajectories.  Since trajectories are continuous curves in a compact domain, they may contain {\em limit sets}, that is, sets of points that are limits of convergent subsequences.  Two particularly important, and easy to describe,
types of limit sets are {\em fixpoints} (roots of $f(x)$) and {\em limit cycles} (closed trajectories that capture periodic behavior).  But there are many other kinds of limit behaviors in dynamical systems --- including the aptly named {\em strange attractors} \cite{Hirsch2004differential}.   The study of the {\em unpredictability} of dynamical systems, also known as {\em  Chaos Theory} \cite{Strogatz2001} is essentially the study of very complex types of limit sets.

However, {\em there is no chaos in two-dimensional dynamical systems,} and the intuitive reason is planarity: trajectories cannot cross, and therefore they ``confine'' one another into benign behavior.  The rigorous statement to this effect is an important result dating back to the end of the 19th century, first stated by Poincar\'e and later proved in its generality by Bendixson \cite{Poincare,Bendixson}.  

\medskip\noindent {\bf Poincar\'e--Bendixson Theorem:}
{\em In a two-dimensional dynamical system $\dot x = f(x)$ on a compact domain where $f$ is continuously differentiable and has no fixpoints, all limit points lie on limit cycles.}

\medskip
\noindent
This theorem not only has many important implications in  physics and biology,\footnote{For instance, it implies the existence of stable oscillations of the van der Pol oscillator \cite{Strogatz2001} which arises naturally in electrical circuits, modelling neurons and seismology. Another application is to glycolysis models: the process living cells use to obtain energy by breaking down sugar molecules, see \cite{hess1979glycolytic,Strogatz2001}.} it is also important in mathematics,\footnote{For instance the Brouwer Fixed Point theorem in two dimensions can be obtained as a corollary.} see  \cite{PBSurvey}. 
Interestingly, it has variants which hold in higher dimensions under suitable assumptions. For instance,  when the system is ``effectively two-dimensional'', or of bandwidth two; i.e., of the form $\dot{x_i}=f_i(x_i,x_{(i-1)\bmod n}), i = 1,\ldots,n$, when the $f_i$s are monotone \cite{MalletParet}.  Such systems have been proposed as models of the origins of life by Eigen and Schuster \cite{eigen1979hypercycle} and, in fact, are what originally interested us in this problem. We explain this connection Section~\ref{sec:discussion}.\footnote{Incidentally, another interesting special case of multidimensional dynamical systems with (almost certain) periodic behavior arises in Chazelle's work on influence systems \cite{Chazelle12}.} 

\subsection{Our Contributions}

{\em In this paper we consider the Poincar\'e-Bendixson theorem from the viewpoint of computation.}  Suppose that we are given a  two-dimensional dynamical system over a compact domain, which is guaranteed to not have a fixpoint (several classes of natural systems are guaranteed to have none, and, in many others, subdomains in which there are no fixpoints can be identified). 
How difficult is it to find a point on a limit cycle?  Or to tell if a given point lies on one?

We first look at these questions in a discrete planar domain: a grid of points, where the dynamical system is an implicit map from each grid point to one of its eight neighbors (grid points at $\ell_\infty$ distance one) such that no two edges cross.  In such discrete dynamical systems, it is clear that the limit cycles correspond to the ``sink cycles'' of the directed graph, and hence a discrete version of the  Poincar\'e-Bendixson theorem trivially holds.  Computationally,  we can show the following:
 
\medskip\noindent {\bf Discrete Poincar\'e--Bendixson Theorem (Theorem \ref{thm:discrete-search}, Theorem \ref{thm:discrete-decision}):}
{\em Given a polynomially computable non-crossing function on a finite subset of the planar grid  
which has no fixpoints, a cycle always exists, but finding it is {\bf PSPACE}-complete.}

\medskip \noindent 
Returning to the continuous domain, given a two-dimensional dynamical system $\dot x = f(x)$ in a compact domain which has no fixpoints, we want to find a limit cycle.  To make the question not a priori impossible, we assume that the function $f$ is Lipschitz continuous and polynomially computable. 
One way to approach this question is by using the framework of ``black box'' complexity of real functions, 
as pioneered in the 1980s by Ko~\cite{Ko}.  Unsurprisingly, here we have a stark impossibility result, which is a corollary of the intractability of finding and testing roots of real functions~\cite{Ko}.

\medskip\noindent {\bf Poincar\'e--Bendixson Impossibility Theorem (Theorem \ref{thm:cts-uncomputable}):}
{\em The problem of finding a point on a limit cycle of $\dot x = f(x)$ in $[0,1]^2$, or the problem of determining if a given point is at most $\delta>0$ away from a limit cycle, with black box access to a Lipschitz and continuously differentiable $f$, has arbitrarily high complexity.}

\medskip 
\noindent 
Finding a fixpoint, if it does exist, is similarly uncomputable.  These results are proved in an appendix.  

Two avenues suggest themselves for getting around this negative result:  either ``look inside the black box'' that computes $f$, or settle for an approximate notion of a limit cycle; {\em we take both.}  First, an {\em approximate} limit cycle would be an orbit that ``comes close to itself'' --- but this is tricky to define: trivially, if we travel infinitesimally, we are close to the point from where we started, but obviously this is not what we mean. A better idea would be to demand intermediate points with all possible gradients, but this confines us to two dimensions. Our definition of an  $\eps$-cycle (see Section \ref{sec:approximate-cycle}) requires that the orbit starting at point $x$ intersects the normal hyperplane to $f(x)$ at $x$ within a distance $\eps$ to $x$. By choosing $\eps <L\|f(x)\|/2$, where $L$ is the Lipschitz constant, the derivatives at the two points must have positive inner product. This indeed gives something that looks and feels like an approximate cycle (see the figure below the abstract).
We prove the following:

\medskip\noindent {\bf Approximate Poincar\'e--Bendixson Theorem (Theorem \ref{thm:approx}):}
{\em Given a dynamical system $\dot x = f(x)$ in a compact domain {\em of any dimension} where $f$ is $L$-Lipschitz continuous and has no $\epsilon$-fixpoints\footnote{An {\em $\epsilon$-fixpoint} is a point $x$ such that $\|f(x)\|< \epsilon$.} an $\nfrac{\epsilon}{3L}$-cycle exists  in the orbit of every point.} 

\medskip \noindent
The proof parallels the original one \cite{Poincare,Bendixson}, except that compactness arguments are replaced by a volume argument.   Notice the uncanny similarity with the statement of
the Poincar\'e-Bendixson theorem stated above --- except for the fact, interesting in itself, that approximation blunts the distinction between two-dimensional and higher-dimensional systems.   
\def\pos{\hbox{\sc sgn}}

How hard is it then to identify $\epsilon$-cycles?  We study this problem in a framework of {\em arithmetic circuits}, with arithmetic operations such as real addition, multiplication and sign as gates, akin to the ones used in the study of the complexity of fixpoints \cite{DGP,CD,EY}. Such a circuit divides the domain into exponentially many cells and encodes a polynomial in each of the cells. We make the natural  assumption that each such polynomial has degree which is a polynomial in the overall size of the circuit. We can show:

\smallskip\paragraph{\bf Complexity of $\epsilon$-cycle (Theorem \ref{thm:eps-complexity})}  {\em Given  $\eps,L>0,$  an $L$-Lipschitz dynamical system through an arithmetic circuit, and a point $x$, determining whether $x$ lies on an $\nfrac{\eps}{L}$-cycle, or finding a point that does, is {\bf PSPACE}-complete. 
}

\noindent
The {\bf PSPACE} upper bound involves solving the differential equation in each cell of the domain as a real analytic function through the Cauchy-Kowalevski Theorem, approximating the solution exponentially closely, and using the Existential Theory of Reals, as well as arithmetic polynomial identity testing and root-finding of analytic functions, to carry out the necessary tests.  The lower bound entails implementing the reductions of our Discrete Poincar\'e-Bendixson Theorem as continuously differentiable functions in a way that does not change the limit cycles.  

Two challenging and important problems remain open:  First, in systems with no fixpoints, can the true limit cycle guaranteed by the Poincar\'e- Bendixson Theorem be approached in polynomial space?   Naturally, in such systems our above-mentioned result allows us to find in polynomial space $\epsilon$-cycles for arbitrarily small $\epsilon$, but these may be very far from a true limit cycle.  
And second, in the specific multi-dimensional dynamical system that was proposed by Eigen and Schuster  \cite{eigen1979hypercycle} as a model for the origin of life, can the limit cycle be approached in polynomial time?  

\section{The Discrete Poincar\'e-Bendixson Theorem}\label{sec:discrete}
\def\ff{f_{\phi}}
In this section we consider a discrete version of the Poincar\'e-Bendixson theorem in two dimensions, and establish the computational hardness of the corresponding statement.  

In particular, suppose that we are given succinct access to a directed graph whose vertices are subset of an exponentially dense grid in $[0,1]^2:$  Given the bit representation of a vertex, we can find out in polynomial time ``where to go next from this vertex''.   In order to retain the planar structure of the two-dimensional dynamical systems, {\em we insist that in no square of the grid both diagonals be used,} that is, the flow does not cross itself. The discrete analogue of a limit cycle in this case is a sink strongly connected component\footnote{A sink strongly connected component is a strongly connected component that has no edge leaving it.} 
in the resulting directed graph.  Because of the nature of this directed graph (all out-degrees are one), sink connected components are precisely {\em cycles} in the graph.  To see this, notice that (a) there is no sink node, since every vertex goes somewhere; (b) any  strongly connected component which is not a cycle must contain a node with out degree greater than one; 
and (c) any non-sink cycle must contain a node with out degree greater than one.  
Thus the analog of a limit cycle in the discrete case is {\em a cycle}.  We give a self-contained proof that both problems of interest are {\bf PSPACE}-hard: (1) to tell if a given point in the graph lies on a cycle; and (2) to find a point on a cycle.   

We reduce to these two problems {\sc QuantifiedBooleanFormula} (QBF), the problem of checking whether a given quantified  formula $\mathcal{I}=Q_nx_n\cdots Q_1x_1\phi(x_1,\ldots,x_n)$ on $n$ variables where each $Q_i \in \{\exists,\forall\}$ has a satisfying truth assignments.  
{We  define a graph implicitly given by $\mathcal{I}$ such that it can contain exactly one of two possible limit cycles: one will occur in the case when $\mathcal{I}$ is  TRUE, and the other when it is FALSE.}
(As it turns out, the reduction for the problem (2) above, finding a point in the cycle, is many-to-one.)  

\subsection{The Reduction} Since it is convenient to ensure that there is no fixed point, we choose the domain to be a subset of $[0,1]^2$ which has a {\em hole}; a torus.
In particular, we let 
$$ T \defeq [0,1]^2 \backslash (\nfrac{1}{7}, \nfrac{6}{7})^2.$$ 
We subdivide $T$ into squares. Each square is divided into {\em grid} points and no two squares share grid points. 
In particular all but the {\em core} square have $2(n+1)+1 \times 2(n+1)+1$ grid points,  where $n$ is the number of Boolean variables in $\phi.$ The core square is a grid of $2^n (1+n(n+1)/2)$ rows and $2(n+1)+1$ columns. Thus, each grid point in each of the squares can be described by local coordinates which are a pair of integers $(i,j).$   In other words, when dealing with grid points we shall  omit denominators. 
 This is our set of nodes. 
The edges will be defined implicitly, through an algorithm.  For the given $\mathcal{I},$ we define a function $f_{\mathcal{I}}$ at every point in $T_n$ which has the property that for each $(i,j)\in T_n$, the function $f_{\mathcal{I}}(i,j)\in \{(i+1,j),(i-1,j),(i,j+1),(i,j-1),(i+1,j+1), (i-1,j+1)\}$  We will denote these values by $R, L, U, D, UR, UL,$ respectively meaning right, left, up, down, diagonally up-and-right, and diagonally up-and-left.  To retain the non-crossing property,
 we require that if $f_{\mathcal{I}}(i,j) = UR$ then $f_{\mathcal{I}}(i,j+1) \neq UL$.  
 
We can think of $f_{\mathcal{I}}$ as a directed graph with out-degree 1 on the vertex set $T_n$ with edges $((i,j),f_{\mathcal{I}}(i,j))$.  Notice that $f_{\mathcal{I}}$ has no fixed point (the corresponding graph has no sink vertex), even though it will have sources, vertices with in-degree zero.

The domain $T$ consists of twenty four ($=7^2 - 5^2$) squares of size ${1\over 7}\times {1\over 7}$, and therefore $T_n$ is the union of twenty-four grids; we call these grids $S_1,S_2,\ldots, S_{24}$, clockwise from the top, see Figure \ref{fig:overview}.  All but $S_{19}$ are grids of size $2(n+1)+1 \times 2(n+1)+1.$ In all but three of these, the function $f_{\mathcal{I}}$ will be very simple.  If $(i,j)\in S_k$ for $k\notin \{18,19,20\}$, then $f_{\mathcal{I}}(i,j)$ will describe a \emph{laminar clockwise flow}: e.g.,  if $(i,j)\in S_k$ with $k\in\{23,24,1,2,3\}$ then $f_{\mathcal{I}}(i,j)=R$, for  $k \in \{5,6,7,8,9\}$ then $f_{\mathcal{I}}(i,j)=D$, and if $(i,j)\in S_4$, the upper-right corner square, then $f_{\mathcal{I}}(i,j) = R$ if $i<j$ and $f_{\mathcal{I}}(i,j)=D$ otherwise, effecting a right turn of the flow. 

\begin{figure}[h!]
\subfigure[Flow in the non-core squares]
{\includegraphics[height=6cm]{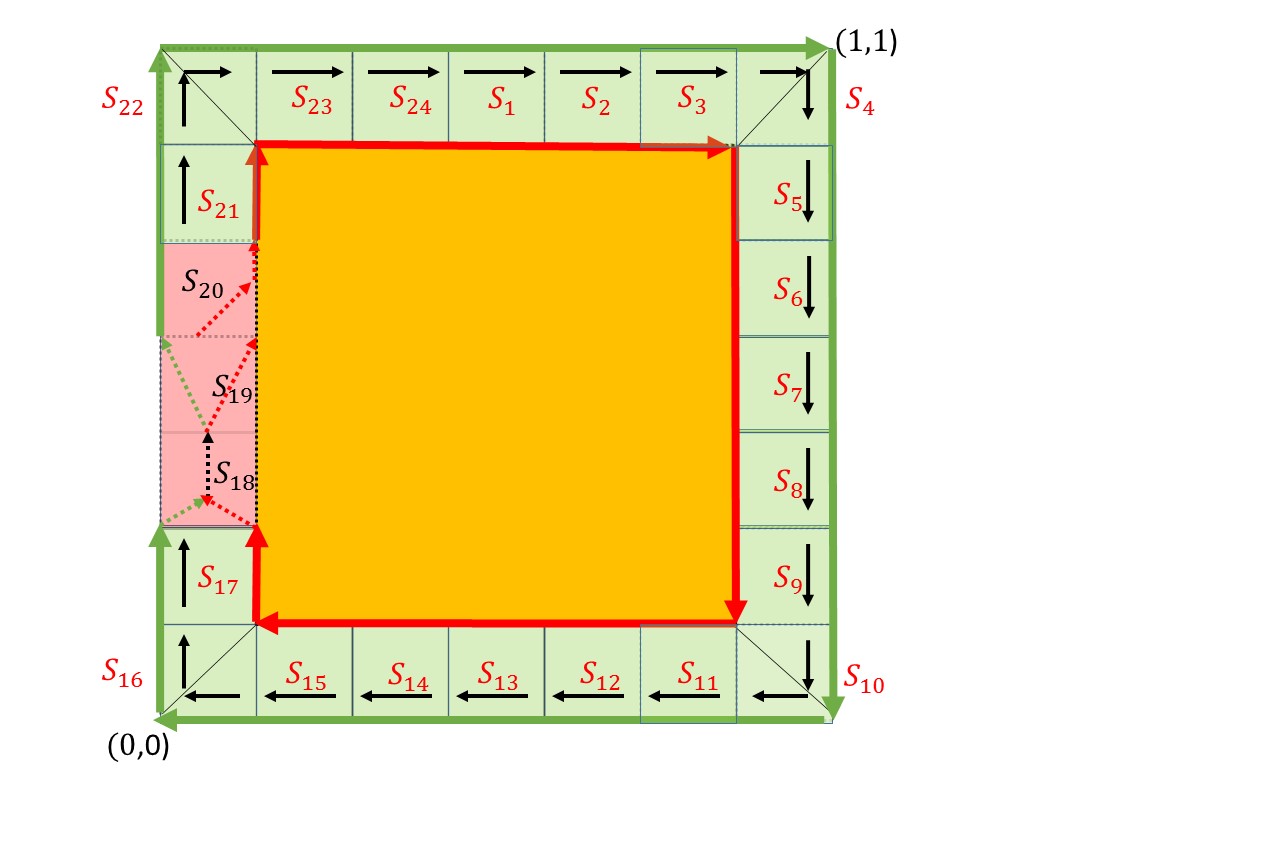}\label{fig:overview}}
\hspace{-2cm}
\subfigure[The single variable case in  $S_{19}$]
{\includegraphics[height=6cm]{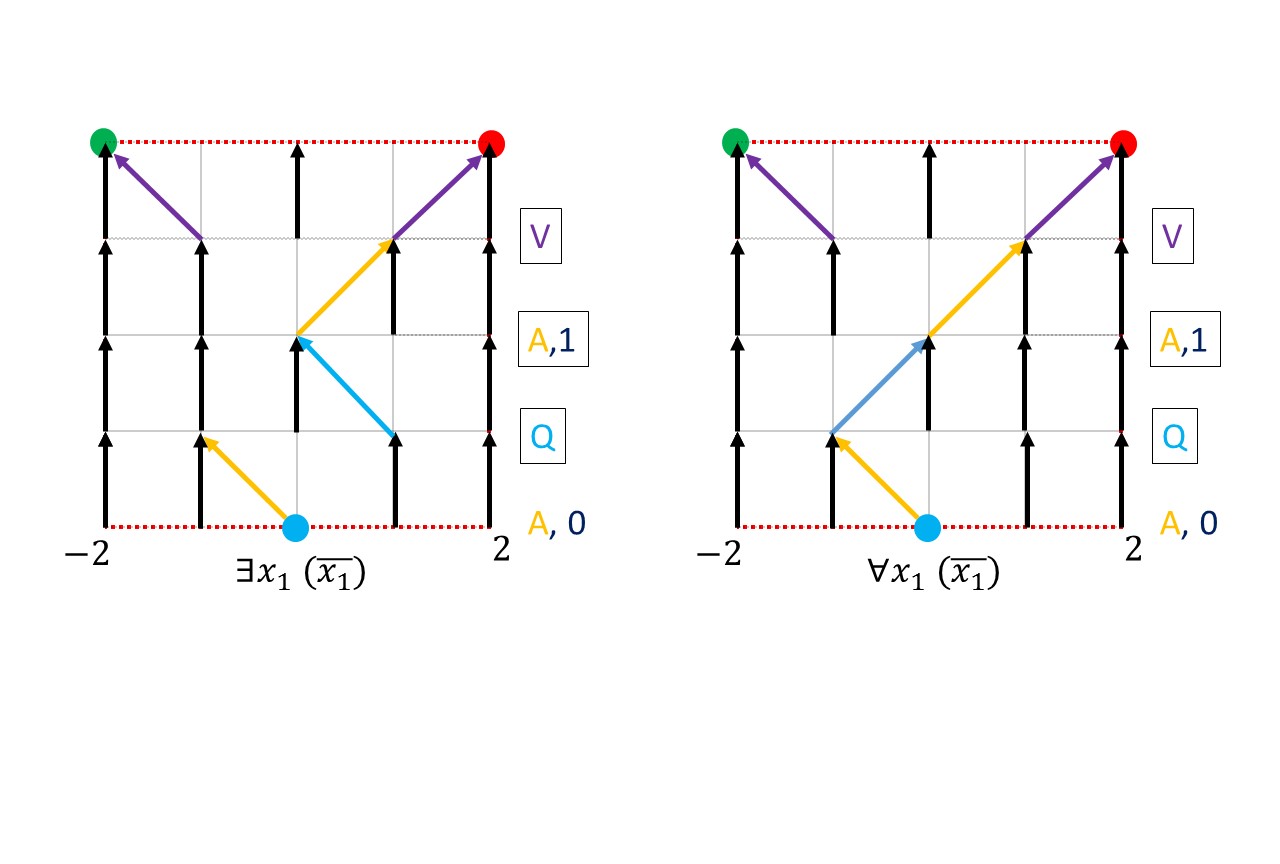}\label{fig:s19}}
\caption{The overall reduction and the core square}
\end{figure}

\begin{figure}[h!]
\subfigure[Recursion when the first quantifier is $\exists$ ]
{\includegraphics[height=6cm]{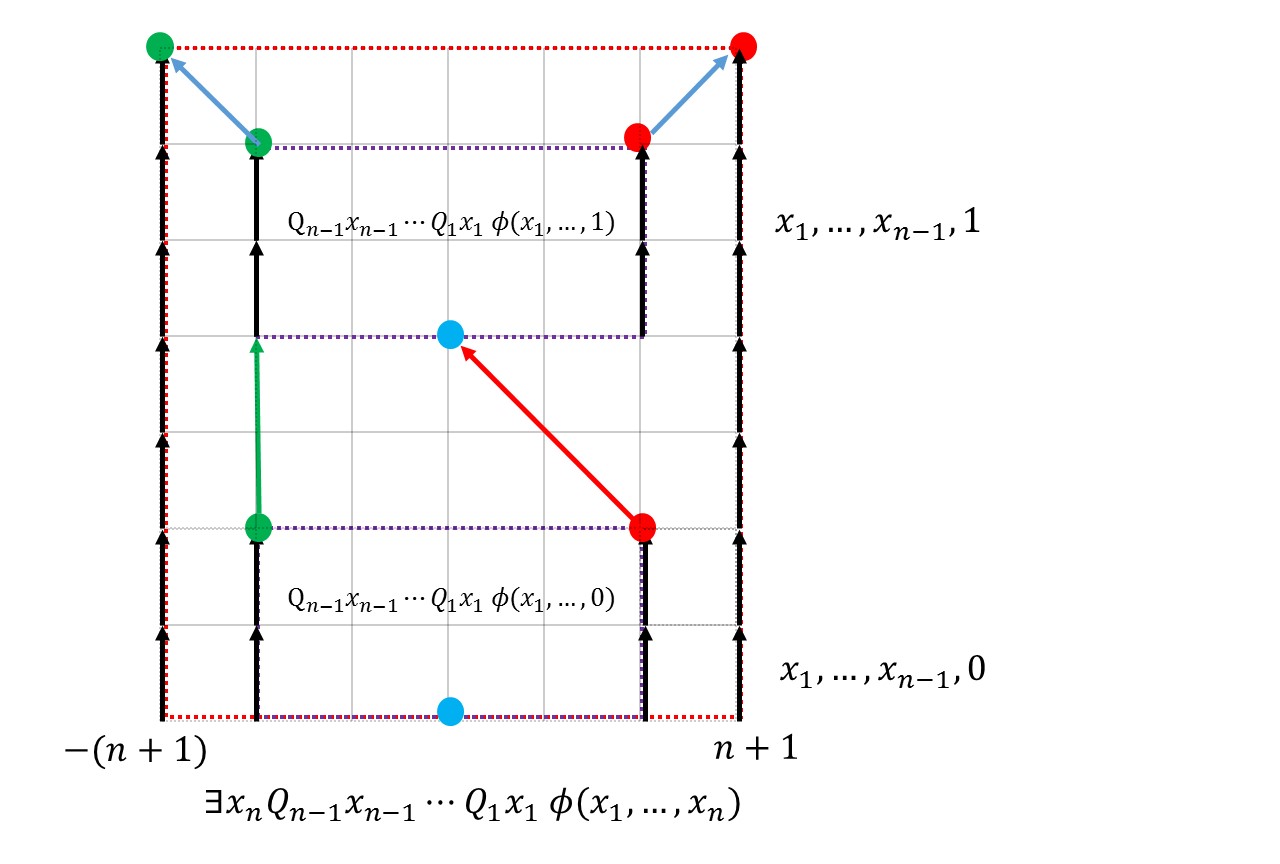}\label{fig:exists}}
\hspace{-2cm}
\subfigure[Recursion when the first quantifier is $\forall$]
{\includegraphics[height=6cm]{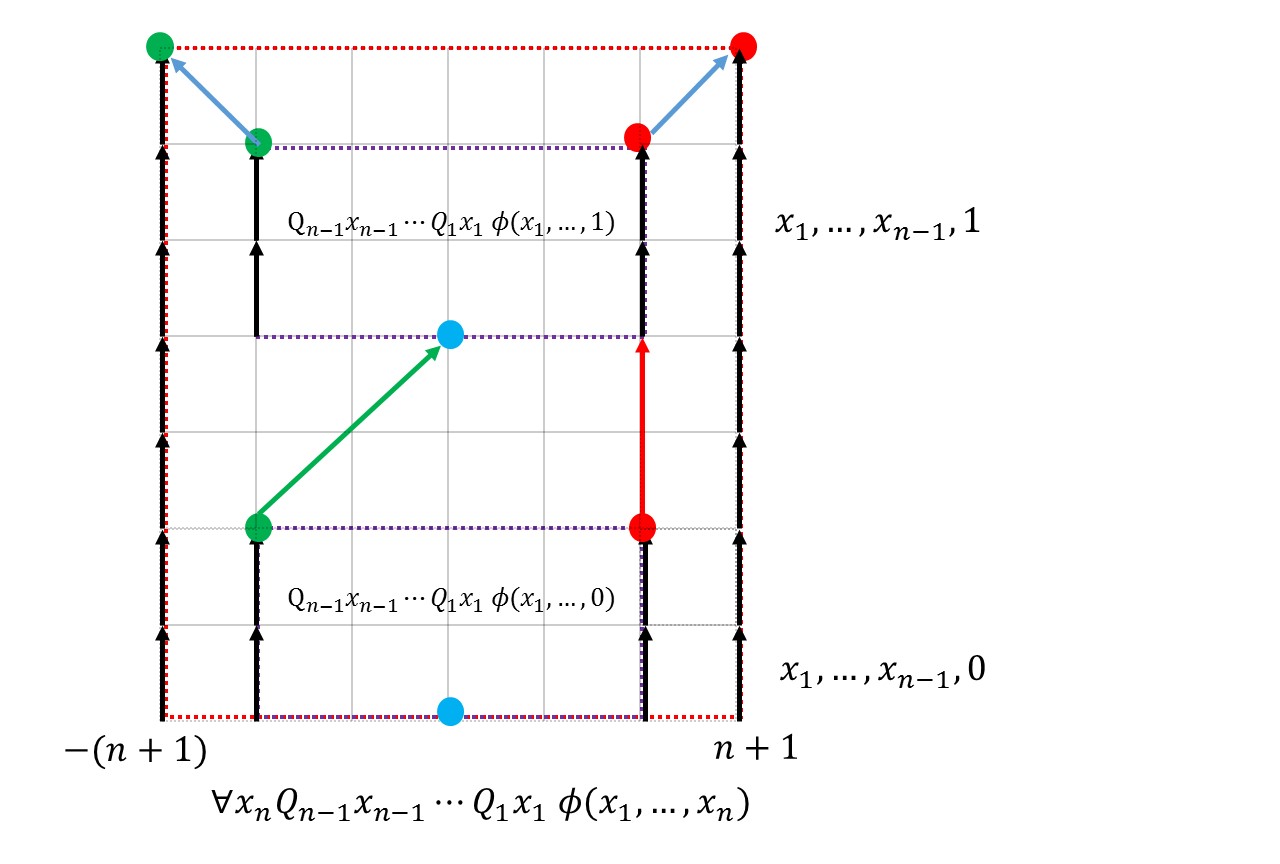}\label{fig:forall}}
\caption{The recursion in $S_{19}$}
\end{figure}

\paragraph{\bf The core square.}
The core of the reduction happens in square $S_{19}$ with some necessary pre- and post-processing  implemented in squares $S_{18}$ and $S_{20}$ respectively. In square $S_{19}$, the values of $f_{\mathcal{I}}$ depend crucially on the QBF instance ${\mathcal{I}}$.  
The total number of rows in the grid $S_{19}$ are $2^n \cdot (1+ n(n+1)/2)$ where there $2^n$ {\em blocks}, one each for an assignment $A=(A_1,\ldots,A_n)$ to $x_1,\ldots,x_n$; we deonte the block corresponding to $A$  by $B_A.$ 
The blocks are arranged according to increasing lexicographic order on bits. Thus, the bottom most  block corresponds to the assignment $\vec{0}=(0,0,\ldots,0)$ and the topmost block corresponds to the assignment $\vec{1}=(1,1,\ldots,1).$  Each block has exactly $1+n(n+1)/2$ rows.
  Each block $B_A$ has $1$ row for the assignment $A$; we call this row $R_{A,0}.$ For each $1 \leq i \leq n,$ there is a block of $i$ rows denoted by $R_{A,i}$ in $B_A.$ Index the rows of $R_{A,i}$ by $\{0,1,\ldots, i-1\}$ and denote the $k$-th row by $R_{A,i,k}.$ 
 There are $2(n+1)+1$  columns in $S_{19}$ that are indexed by integers from $-(n+1)$ to $n+1.$ Thus, $S_{19}$ has $2^n \times (1+n(n+1)/2)$ rows and $2(n+1)+1$ columns.
 
 \paragraph{\bf The flow in the core square.}    
   Our intention  is to use  $\mathcal{I}=Q_nx_nQ_{n-1}x_{n-1}\ldots Q_1x_1 \phi$ to define a flow such that if we start tracing it from the middle column of the bottommost row of $S_{19},$ ($R_{\vec{0},0},0),$ we end up at either the left corner of the topmost row of $S_{19}$ or the right corner of the topmost row depending on whether $\mathcal{I}$ is TRUE or FALSE respectively. Thus, the flow {\em runs} over all the $2^n$ assignments, and ensures that the quantifiers are satisfied and by the time it is out of $S_{19}$ it knows whether $\mathcal{I}$ is TRUE or FALSE.  
   
   All the flow from below $S_{19}$ is routed to $(R_{\vec{0},0},0).$ Thus, no matter where one starts from, we are brought back to $(R_{\vec{0},0},0).$ To ensure the condition on the flow in the core square, we use recursion. For an assignment $A,$ let $A_n$ be its most significant bit and let $Q_n$ be the quantifier corresponding to $x_n.$ Based on $A_n$ we can divide the set of blocks into two contiguous set: $\{B_{0A'}\}_{A' \in \{0,1\}^{n-1}}$ and $\{B_{1A'}\}_{A' \in \{0,1\}^{n-1}}.$ We would like to ensure recursively that if we start at $(R_{0\vec{0},0},0)$ then if the quantifier $Q_n$ is $\exists$ we end up at $(R_{1\vec{0},0},-n)$ if $Q_{n-1}x_{n-1}\cdots Q_1x_1 \phi(x_1,\ldots, x_{n-1},0)$ is TRUE and at $(R_{1\vec{0},0},0)$ if FALSE (see Figure \ref{fig:exists}). On the other hand if $Q_n=\forall,$ then starting at $(R_{0\vec{0},0},0),$ we end up at $(R_{1\vec{0},0},0)$ or $(R_{1\vec{0},0},n)$ depending on whether $Q_{n-1}x_{n-1}\cdots Q_1x_1 \phi(x_1,\ldots, x_{n-1},0)$ is TRUE or FALSE respectively (see Figure \ref{fig:forall}). The base case is illustrated in Figure \ref{fig:s19}.

  We now define the function $f_{\mathcal{I}}$ in $S_{19}$ formally. For the grid points in $S_{19}$ not covered in the below cases, the default flow is $U.$    
 
 \begin{itemize}
 \item {\bf $R_{A,0}$-rule.} For each $A\in \{0,1\}^n,$ $f(R_{A,0},0)=UL$ if $A$ satisfies $\phi$ and $f(R_{A,0},0)=UR$ if $A$  does not satisfy $\phi.$
 
 \item {\bf $R_{A,i}$-rule.} We first associate an appropriate {\em type} $Q$ (for quantifier), $V$ (value), $E$ (empty) to each block of rows $R_{A,i}$ as follows.
 
 \begin{enumerate}
 \item IF $A_i=0$ and $\forall l<i$ $A_l=1,$ THEN type is $Q$ and further
 \begin{enumerate}
 \item IF $Q_i=\forall$ THEN $f_{\mathcal{I}}(R_{A,i,m},-i+m)=UR$ for $m \in \{0,1,\ldots, i-1\}.$ 
 \item  ELSE IF $Q_i=\exists$ THEN $f_{\mathcal{I}}(R_{A,i,m},i-m)=UL$ for $m \in \{0,1,\ldots, i-1\}.$ 

 \end{enumerate}
 \item IF $A_i=1$ and $\forall l<i$ $A_l=1,$ then type is $V$ and, $f_{\mathcal{I}}(R_{A,i,0}, -i)=UL$ and $f_{\mathcal{I}}(R_{A,i,0}, i)=UR.$
 \item ELSE  type is $E$ and the default $U$ flow is used everywhere. 
 \end{enumerate}
 \end{itemize}

\noindent
The following lemma is now evident.
 \begin{lemma}
 Starting at $R_{\vec{0},0}$ the trajectory goes to the upper left corner of $S_{19}$ if the QBF is TRUE and to the upper right corner of $S_{19}$ if the QBF is FALSE.
  \end{lemma}

 \begin{figure}
\hspace{1cm}
\subfigure[]{\includegraphics[height=5cm]{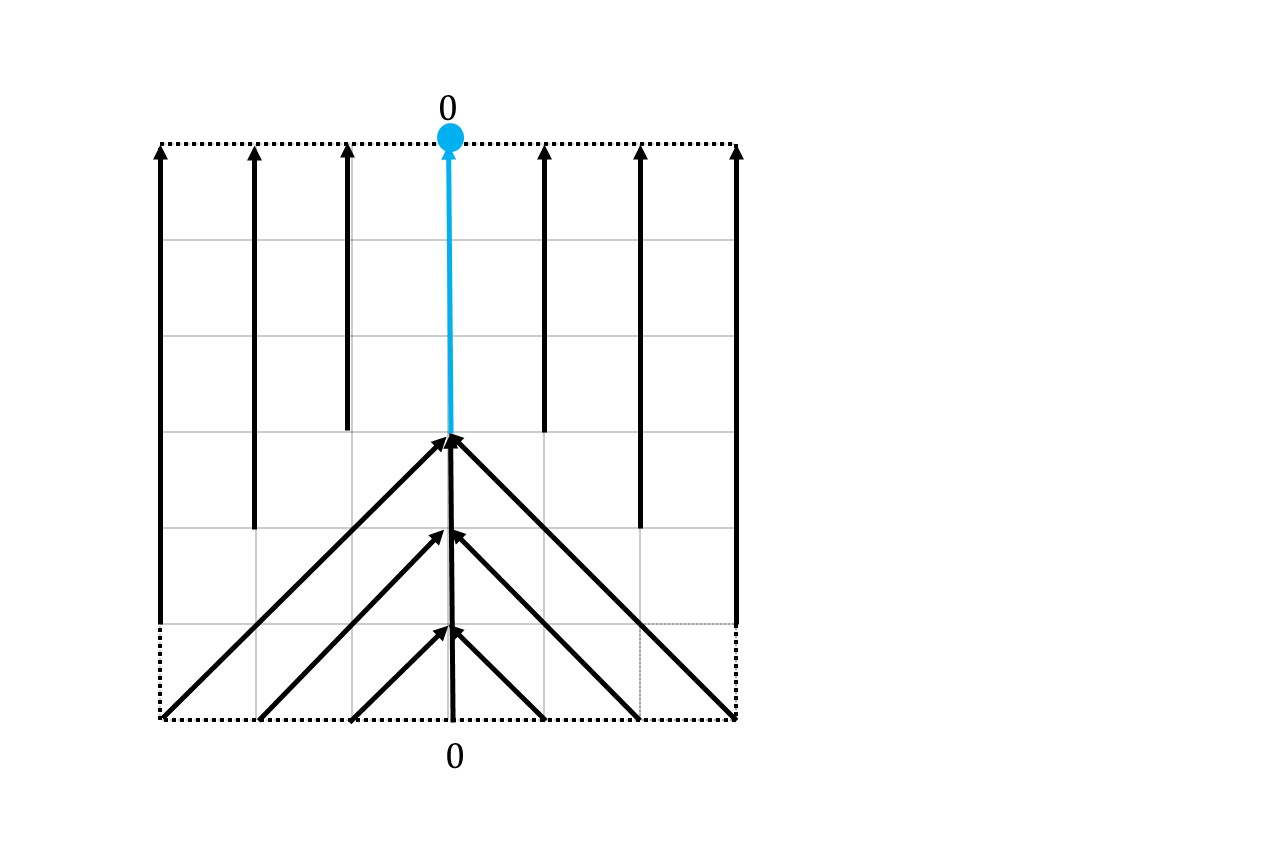}\label{fig:s18}}
\subfigure[]{\includegraphics[height=5cm]{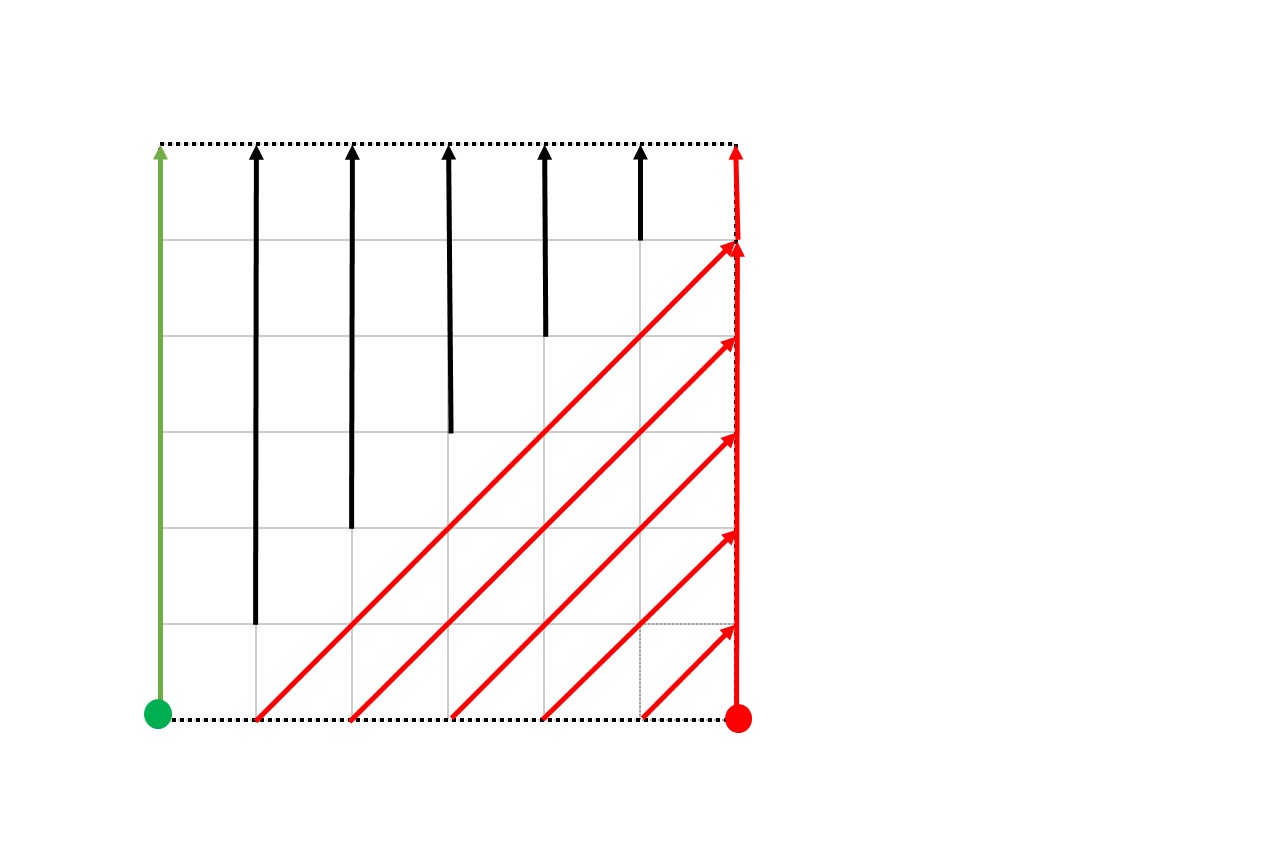}}
\caption{The left figure is $S_{18}$ for pre-processing and the right figure is $S_{20}$ for post-processing.}\label{fig:s18-20}
\end{figure}

\noindent
Now, define {\sc Discrete Limit Cycle} to be the following decision problem:  given $n$, a Boolean circuit $C$ that computes a displacement function $f(i,j)$ for point $(i,j)\in T_n$, and a query point $P=(i_p,j_p)\in T_n$, does $P$ lie on a cycle of the graph on $T_n$ defined by $C$?  It should be clear that the function $f_{\mathcal{I}}$ as defined in our reduction can be implemented by a Boolean circuit of size polynomial in $n$, given access to the quantifiers and a Boolean circuit for $\phi$.  Hence, through the reduction described above, we have shown the following (membership in {\bf PSPACE} follows from the fact that finding the strongly connected components of a graph is in logarithmic nondeterministic space):

\begin{theorem}[Discrete Poincar\'e-Bendixson--decision  version]\label{thm:discrete-search}
The problem {\sc Discrete Limit Cycle} is  {\bf PSPACE}-complete.
\end{theorem}
\noindent
In a dynamical system, the question of interest is not so much to tell if a point is on a limit cycle, but to {\em find} a point on a limit cycle, and our analysis so far has not given us clues about the complexity of this problem.   Let us define {\sc Point on Discrete Limit Cycle} to be the following computational problem:  given $n$ and a Boolean circuit $C$ as above, output a point that lies on a discrete limit cycle of the resulting discrete dynamical system.  We prove the following theorem:

\begin{theorem}[Discrete Poincar\'e-Bendixson--search version]\label{thm:discrete-decision}
The problem {\sc Point on Discrete Limit Cycle} is  {\bf PSPACE}-complete.
\end{theorem}

\begin{figure}[h!]
\subfigure[Modified $S_{18}$]
{\includegraphics[height=4cm]{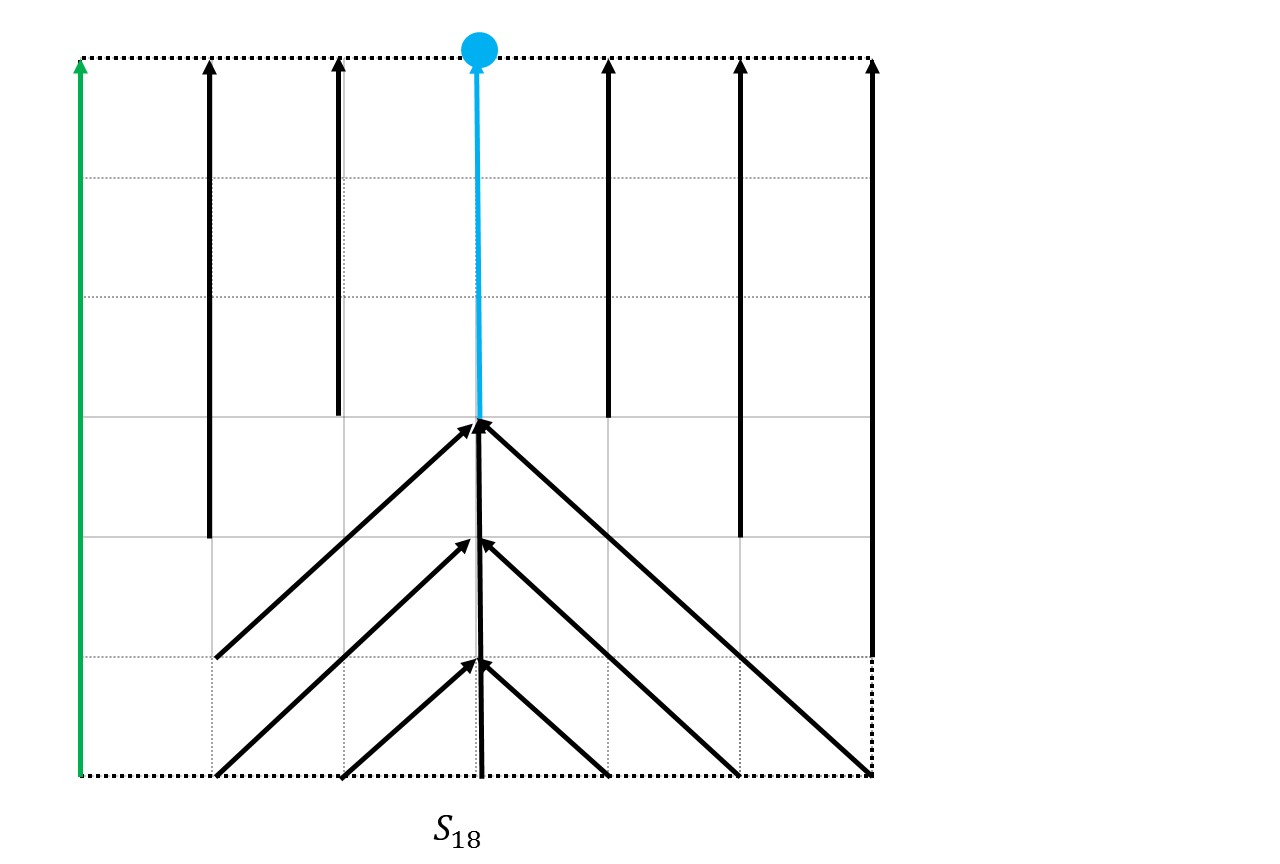}\label{fig:S18new}}
\hspace{-1cm}
\subfigure[$S_{18}'$]
{\includegraphics[height=4cm]{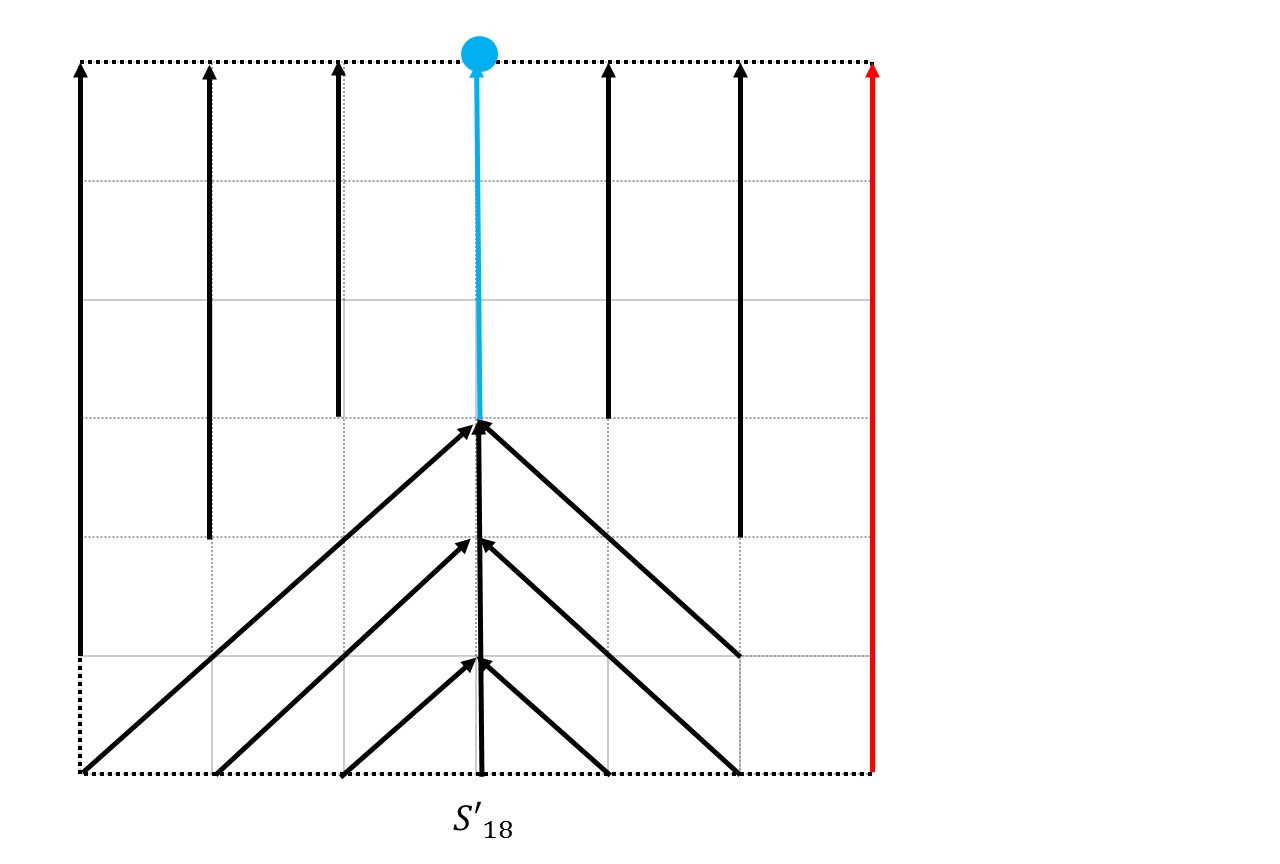}\label{fig:S18prime}}
\hspace{-1cm}
\subfigure[The overall modification]
{\includegraphics[height=4cm]{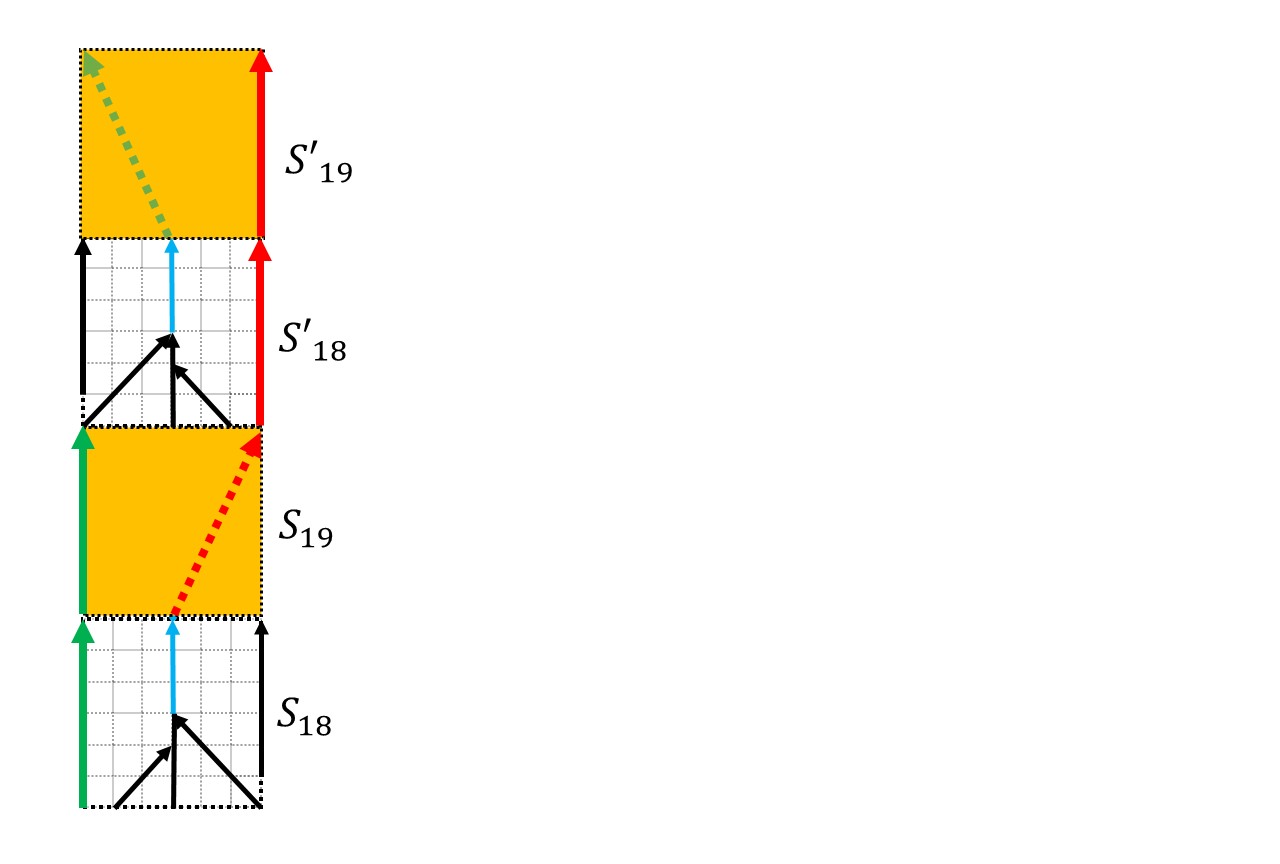}\label{fig:modified}}

\caption{The modifications for {\sc Point on Discrete Limit Cycle}}
\end{figure}

\begin{proof}
Suppose that we have an algorithm $A$ which, given any circuit $C$ that computes a displacement function on $T_n$, identifies a point $A(C)=(i,j)\in T_n$ which lies on a limit cycle.  We shall use this algorithm to solve QBF.   
Given a QBF $\mathcal{I}$ we construct a circuit implementing the dynamical system associated with $\mathcal{I}$ in a way almost exactly the same as in the previous reduction, except that,
we create two slightly modified copies of $S_{18}$  ($S_{18}$ and $S'_{18}$ in Figure \ref{fig:S18new} and \ref{fig:S18prime} respectively) and two identical copies of $S_{19}$ ($S_{19}$ and $S'_{19}$).  One can compare these with the $S_{18}$ in the previous proof: the first one has the property that the bottom left corner is connected to the top left and the other has the property that the bottom right corner is connected to the top right corner. 
Together, $S_{18}, S_{19}, S_{18}', S_{19}'$ replace $S_{18}$ and $S_{19}$ in the previous proof after scaling them down vertically by a factor of $2$; see Figure \ref{fig:modified}. 

We then apply algorithm $A$ on the resulting discrete dynamical system.  Algorithm $A$ will return a point $P$ on a limit cycle. As before,  our reduction ensures that  limit cycle of this system is unique, and its nature is very concrete and depends crucially on the outcome of the question whether $\mathcal{I}$ is TRUE or FALSE: 
\begin{itemize}
\item In all squares of $T_n$ except for $S_{18},S_{18}',S_{19}, S_{19}'$ the limit cycle is on the outer boundary of $T_n$ if the outcome is TRUE, and the inner boundary otherwise.
\item If the returned point is in one of $S_{18},S_{19},$ then we check if it is on the left boundary in which case the outcome is TRUE else FALSE. 
\item If the returned point is in one of $S_{18}',S_{19}',$ then we check if it is on the right boundary in which case the outcome is FALSE else TRUE. 
\end{itemize}
\end{proof}

\section{$\eps$-Cycles and the Approximate Poincar\'e-Bendixson Theorem}\label{sec:approximate-cycle}

In this section we  prove that, with  the right notion of an approximate cycle,  a version of the Poincar\'e-Bendixson theorem holds in arbitrary dimensions:  in the absence of approximate fixpoints, certain orbits come very close to forming a cycle 
(we say that $x\in T$ is an {\em $\epsilon$-fixpoint} if $\|f(x)\|<\epsilon$). 

\begin{definition}[$\eps$-cycle]
Let $\epsilon>0$, and let $\dot x = f(x)$ be a dynamical system $\dot x = f(x)$ on a compact subset $T$ of $\mathbb{R}^d$ where $f$ is $L$-Lipschitz continuous.    We say that a trajectory starting at $x(0)$ is an  {\em $\epsilon$-cycle} if for some $t>0,$ $x(t)$ is in the $(d-1)$-dimensional ball of radius $\eps$ centered at $x(0)$ and is orthogonal to $f(x(0))$, and $\langle f(x(0)),f(x(t))\rangle>0$.  
\end{definition}

\noindent
To understand the last requirement, compare the two parts of Figure \ref{fig:elc}. We would like to call the first one an approximate limit cycle, and not the second one.  As we shall see, this requirement is {\em redundant} when $\epsilon<\|f(x(0)\|/L$, because of the Lipschitz condition.

\begin{figure}[h]
\subfigure[$\eps$-cycle]{\includegraphics[height=4.8cm]{elc.jpg}} 
\hspace{1cm}
\subfigure[Not an $\eps$-cycle]{\includegraphics[height=4.8cm]{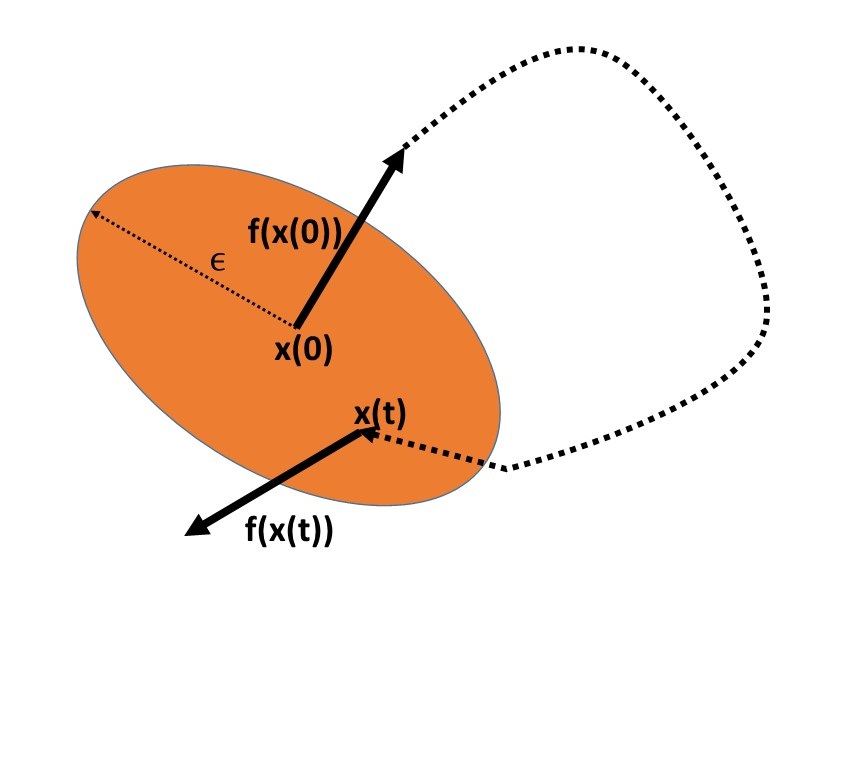}} 
\caption{$\eps$-Cycle}\label{fig:elc}
\end{figure}

 \begin{theorem}[Approximate Poincar\'e - Bendixson Theorem]\label{thm:approx}
For any  sufficiently small $\epsilon>0$, any dynamical system $\dot x = f(x)$ on $[0,1]^d,$ $d\geq 2$ with $f$  $L$-Lipschitz continuous, either has an $\epsilon$-fixpoint, or an $\nfrac{\epsilon}{ L}$-cycle. Further, the length of the $\nfrac{\eps}{ L}$-cycle is $O\left(\nfrac{L}{\eps} \right)^d.$
\end{theorem}

\begin{proof}
Assume for the sake of contradiction that $f$ does not have an $\eps$-fixpoint, nor an $\nfrac{\eps}{ L}$-cycle.  Thus for all $x \in T,$ $\|f(x)\| \geq \eps.$ Let $B(x,\delta)$ denote the $d-1$-dimensional ball of radius $\delta$ centered at $x$ and orthogonal to $f(x).$  
Suppose we start with a point $x(0)$ and continue for time $T$ to reach $x(T)$.  In the absence of fixpoints and limit cycles, all $x(t)$'s on this orbit are distinct.   If there are times $0\leq t_1<t_2\leq T$ such that $B(x(t_1),\delta)$ and $B(x(t_2),\delta)$ intersect, where $\delta={\nfrac{\eps}{ 3L}}$, then we claim that the orbit from $x(t_1)$ to $x(t_2)$, is an $\nfrac{\epsilon}{L}$-cycle.  It is certainly the case that the two endpoints are within $\nfrac{2\eps}{3L}$ of each other, by the triangle inequality.  To see that $f(x(t_1))$ and $f(x(t_2))$ have positive inner product, suppose that they do not.  Then, since $\|f(x(t_1))\|,\|f(x(t_2))\| \geq \eps$, $\|f(x(t_1))-f(x(t_2))\| > \eps > L\|x(t_1)-x(t_2)\|$, violating the Lipschitz condition.  Finally, it is easy to see that there is a point $x(t_2')\in B(x(t_1),\nfrac{\eps}{ L}$ within $O(\nfrac{\eps}{L})$ of $x(t_2)$.

Thus, for $f$ not to have an $\eps$-cycle, the $B(x(t),\delta)$'s are disjoint  for all $0 \leq t\leq T.$  
It follows that the volume of the body obtained by sweeping an $d-1$-dimensional ball of radius $\nfrac{\eps}{3L},$ normal to the trajectory $(x(t))_{t=0}^T$ is lower bounded by 
\begin{eqnarray*}
  (1-\eps) V_{d-1}(\nfrac{\eps}{3L}) \cdot \int_{0}^T \| \dot{x}(t)\| dt &=&  (1-\eps)  V_{d-1}(\nfrac{\eps}{3L}) \cdot  \int_{0}^T \| f(x(t))\| dt \\
  &\geq & (1-\eps) V_{d-1}(\nfrac{\eps}{3L}) \cdot \eps \cdot T \\ 
  &\gtrsim &(1-\eps)\frac{1}{\sqrt{(d-1)\pi}} \left(\sqrt{ \frac{2\pi e}{d-1}}\right)^{d-1}\left(\frac{\eps}{3L} \right)^{d-1} \cdot \eps  \cdot T.
 \end{eqnarray*}
 Here we have used several facts:
 \begin{enumerate}
 \item For every $y$ in the  $d-1$-dimensional ball of radius $\nfrac{\eps}{3L}$  centered at $x(t),$ $\|f(y(t)\| = \|\dot{y}(t)\| \geq (1-\eps) \|f(x(t)\|.$ This allows us to change the order of integration and obtain the first term. 
\item   The length of curve traced by $x(t)$ from time $0$ to $T$ is $\int_{0}^T \|\dot{x}(t)\| dt.$ 
\item $V_{d-1}(r),$  the volume of the $d-1$-dimensional ball of radius $r$ and it is 
 $ \frac{1}{\sqrt{(d-1)\pi}} \left(\sqrt{ \frac{2\pi e}{d-1}}\right)^{d-1}r^{d-1}$ up to a factor of $1+O(d^{-1})$ via the Stirling approximation.
\end{enumerate}
But the volume of this body obtained by sweeping has to be at most $(1+{\eps\over L})^d$ as it is contained in $[-{\eps\over 3L},1+{\eps\over 3L}]^d$ (recall that we have assumed our domain is the $d$-cube). Thus, starting at any $x(0)$, within 
 $T \lesssim {\sqrt{(d-1)\pi}} \left(\sqrt{ \nfrac{d-1}{2 \pi e}}\right)^{d-1}\left(\nfrac{(3L)}{\eps} \right)^{d-1} \cdot \nfrac{(1+\eps/L)^d}{(\eps (1-\eps))}$ 
 we obtain an $\eps$-cycle. This proves the theorem.  
\end{proof}

\medskip\noindent
Note that  this result does not require the differentiability of $f$ and, more importantly, holds for {\em all} dimensions.  

\section{The Complexity of Approximate Cycles}\label{sec:circuit}
How hard is it to actually {\em find} an $\eps$-cycle?  Here we show that the problem is intractable.  We first fix our computational model, which captures the intended generality of computation:
 {\em The arithmetic circuit model,} used, e.g., in \cite{DGP,CD,EY} in the study of fixpoint problems.  
 We assume that the functions that define the dynamical system $\dot x =f(x)$ are given by a circuit $C$ with $d$ input variables and $d$ output variables ($d$ is the dimension, and is assumed to be fixed).  The gates of the circuit are in the basis $\{+, -, *, \pos\}.$   Here $\pos(x) = 1$ if $x>0$ and $\pos(x) = 0$ otherwise
 \footnote{For concreteness and economy we state our results for this minimal basis; we see no clear obstacles in adding division and even arbitrary analytic functions to the basis.}. We also assume that the circuit has a number of rational constants as additional inputs, whose bit description is part of the circuit's size $|C|$. $C$ partitions the domain into exponentially many ``cells'' (regions in which the assignment of $0,1$ values to the $\pos$ variables is fixed) and encodes a polynomial in each of the cells.  We further assume, as is often assumed in work on arithmetic circuits, see, e.g.,~\cite{ShpilkaY}, that the degree of each such polynomial is $|C|^{O(1)}$.  This assumption, along with the $L$-Lipschitzness of the function computed by this circuit, implies via a simple interpolation argument (such as Lemma 3.4 in \cite{parusinski2013new}) that the coefficients of   the polynomial in each of the cells remains bounded by $|C|^{O(|C|)}L,$ which is at most an exponential in the input size.  This allows us to prove that the radius of convergence of the real analytic functions  is large enough for our method for establishing the upper bound to work. A quantitative version of  the Cauchy-Kowalevski Theorem states that for a $d$-dimensional system $\dot{x}=p(x)$ with the initial condition $x(0)=0$, one can express the solution $x(t)$ around $0$ as a power series whose radius of convergence depends on the properties of $p.$  For instance, if $p(x)=\sum_{\alpha}c_\alpha x^\alpha$ is a multivariate polynomial (which converges everywhere), then the radius of convergence of $x(t)$  is  $R=\max_{r>0} \frac{r}{2d \cdot \max_{\alpha} |c_\alpha|r^{|\alpha|}},$ see Lemma 4.2 and Theorem 4.3 in    \cite{Driver-CK}. Thus, if each $|c_\alpha|$ is at most $|C|^{O(|C|)}L$, then $R$ is, roughly, at least $\frac{1}{d|C|^{O(|C|)}L}.$ Of course, $0$ could be replaced by any point by shifting.

\begin{theorem}[Complexity of an $\epsilon$-cycle]\label{thm:eps-complexity}
Given  $\eps,L>0$ and  a dynamical system through an arithmetic circuit $C$ (as described above) that computes an $L$-Lipschitz continuously differentiable function,  finding an $\epsilon$-fixpoint or a point that lies on an $\nfrac{\eps}{L }$-cycle can  is {\bf PSPACE}-complete in two or more dimensions.
\end{theorem}
\begin{proofof}{Sketch} {\bf Upper bound.} We start at a point $x(0)$ and approximate its trajectory for $T=O\left(\nfrac{L}{\eps} \right)^d$ time until an $\eps$-cycle is formed as guaranteed by the proof of Theorem \ref{thm:approx}.  We proceed from cell to cell, solving approximately  the differential equations in each cell through the Cauchy-Kowalevski Theorem \cite{CK} to obtain the solution of $d$  analytic functions (see \cite{kawamura2010complexity} for a complete argument).  The solution is {\em truncated} at the cell boundaries.  That is, we compute, approximately, the smallest time $t$ at which the trajectory intersects one of the cell boundaries (this can be carried out in polynomial space through, for example, the existential theory of the reals ({\bf ETR}) \cite{ETRCanny}).  
The exponential bound on the coefficients of the polynomials implies that the radius of convergence of the power series that describes the solution trajectories is not less than one over an exponential. Further, if at some point $\|f(x)\| \leq \min\{\eps, 2^{-\mathrm{poly}(|C|)}\},$ then we can declare $x$ to be an $\eps$-fixpoint.  Thus, within polynomial space it is possible to achieve the exponential approximation needed in order for our overall approximate solution to always be within $\nfrac{\eps}{ 2}$ of the true solution throughout the simulation for exponential time $T$.  Identifying the two points $x(t_1)$ and $x(t_t)$ in the proof of Theorem  \ref{thm:approx} is also easy by reusing space. Finally, using {\bf ETR} we can check the promised Lipschitzness and continuous differentiability conditions at each cell. Further, by recovering the circuit that computes the polynomials in two adjoining cells and computing the circuits corresponding to their derivatives \cite{ShpilkaY}, we can check that they agree on the boundary curve (which is the zero set of the {\pos} gate whose sign separates the two cells- and is polynomial degree as well); this step can  use polynomial identity testing which is in {\bf PH}.  This completes our sketch of the {\bf PSPACE} upper bound.

{\bf Lower bound.} For the lower bounds, we show how to modify the reductions to the discrete problems given in Theorems \ref{thm:discrete-search} and \ref{thm:discrete-decision} so that they become reductions to the corresponding continuous problems.  In both cases we convert the discrete function $f_\phi$ from grid points to grid points in those reductions into a continuous function effecting the same flow qualitatively --- and hence the same limit cycles.  

Our domain $T\subseteq [0,1]^2$ is  similar to the discrete one, consisting of $12$ squares (3 squares in each region $s_1, s_3$ and $s_5$ and the 3 individual squares $s_7, s_8$ and $s_9$) joined by four annulus quadrants (see Figure \ref{fig:continuous}).  We first describe the circuit computing the $f(x,y)$ and $g(x,y)$ functions in the 2D dynamic system.  First, the circuit will contain a series of $O(n)$ $\pos$ operations (where $n$ is the number of variables in the original QBF instance) extracting from $(x,y)\in [0,1]^2$, through binary search, the grid square (or annulus quadrant) in which the point $(x,y)$ lies, plus its precise coordinates within this grid square (we shall henceforth call these local coordinates, by notation abuse, $(x,y)$).

The flow in the four annulus quadrants is a circle flow, with  equations $\dot{x} = y, \dot{y} = -x$ (or the three symmetric versions), see Figure \ref{fig:continuous}. 
For the twelve squares, we first note that there are three basic kinds of grid squares, depending on the flow from the two ``base'' grid points (where the flow reaches first):  either they both go U,  both go UL, or the one on the left goes U and the other UL (in addition to the symmetric cases (UL,U), (U, UR), etc.).  In the (U,U) case, the equations are, naturally enough, $\dot{x} = 0$ (we only define $\dot{x}$, since outside the annulus quadrants it is always the case that $\dot{y}=1$, or the three symmetric versions).
  In the (UL, UL) case, we use this function:  $\dot{x} = 6y^2-6y$, which simulates the UL diagonal in a smooth and continuously differentiable way (see Figure \ref{fig:continuous}(b)). 
Further, in the (U,UL) case, we interpolate between these two equations: $\dot{x}= x( 6y^2-6y)$.    The symmetric cases are treated the same way. Finally, since some of our grid ``squares'' are {\em rectangles} with unequal $x,y$ lengths, we can modify these differential equations by scaling the $x$ coordinate by the appropriate factor which is about $2^n.$ 
  \begin{figure}[h]
\caption{The continuous reduction}\label{fig:continuous}
\includegraphics[height=12cm,width =15cm]{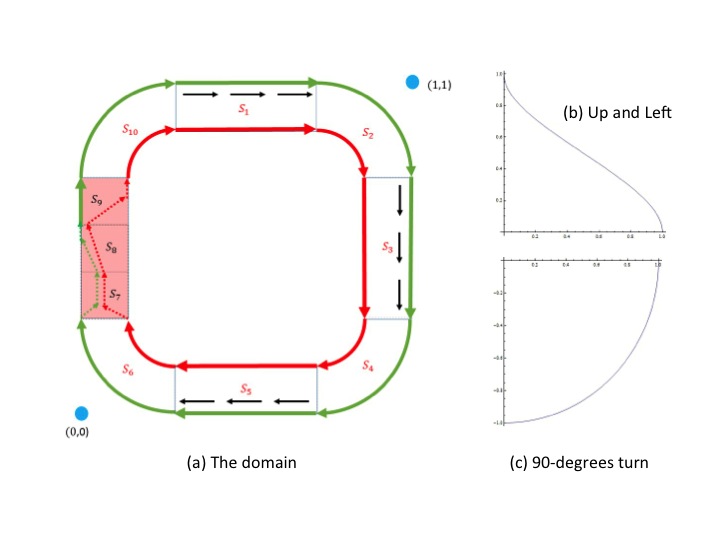}
\end{figure}

This completes the definition of the flow in the domain $T$.  It is easy to check that the function is continuously differentiable, has Lipschitz constant $O(2^n)$, and has the same limit cycle as its discrete counterpart.  {\bf PSPACE}-completeness follows. It is also easy to see that the function described above can be computed by an arithmetic circuit from the family used for the upper bound. In fact the polynomials in the construction just described have {\em constant} degree  and all {\pos} gates apply to only {\em linear functions} in the input variables. %
Finally, for three or more dimensions, the simulation of the multi-dimensional flow is done in an analogous way, whose details are omitted. 
\end{proofof}
As a corollary, if the hypothesis of the Poincar\'e--Bendixson Theorem (no fixpoint) holds, in polynomial space we can find cycles that come arbitrarily near closing.

\section{Discussion and Open Problems}
\label{sec:discussion}
We pointed out that the famous Poincar\'e--Bendixson Theorem for 2D dynamical systems has an interesting computational life: finding a limit cycle is intractable when recast in a discrete setting.   For the continuous problem we show an approximate version of the theorem:  trajectories that come $\epsilon$ near closing exist whenever the norm of the driving function is bounded below.  Finding such an approximate cycle is again {\bf PSPACE}-complete in the arithmetic circuit model.  Interestingly, unlike the Poincar\'e--Bendixson Theorem, the approximate version  holds in three or more dimensions.
\paragraph{\bf What is the complexity of approaching a true limit cycle if there are no fixpoints?}   This is the important problem left open here.  Our results trivially imply that it is no easier than {\bf PSPACE}, but can it be done in {\bf PSPACE}?  (The absence of fixpoints can be checked through {\bf ETR}).  Naturally, in the absence of fixpoints we can find in {\bf PSPACE} $\eps$-cycles for arbitrarily small $\eps>0$, but these may be far from the true limit cycle. One possible approach is this:  even though the approximate Poincar\'e--Bendixson Theorem holds in all dimensions, the 2D case is still special, because it provides clues about the true limit cycle:  {\em It lies inside the $\eps$-cycle.}  The trajectory between $x(t_1)$ and $x(t_2)$, together with the $x_1-x_2$ segment, divide the domain in two parts, of which ``inside'' is the one pointed to by $f(x(t_2))$, and a limit cycle is guaranteed to exist in there.  What is an appropriate notion of ``progress'' through which, if we continue finding $\eps$-cycles in this subdomain, we will eventually come close to the true limit cycle?  Alternatively, can the true limit cycle be somehow found through {\bf ETR}, by exploiting the cell structure of the domain and the polynomial nature of the driving functions?  As a caution to optimism here, we have also shown in the appendix that in the ``black box'' model of \cite{Ko}, in which a Turing machine computes bits of the result through oracle calls to polynomial-time Turing machines computing the driving functions, true limit cycles are uncomputable.  

\paragraph{\bf The origins of life question.}
Our original motivation for looking at the Poincar\'e-Bendixson theorem was the influential work of Eigen and Schuster \cite{eigen1979hypercycle} who considered the following dynamical system over the $n$-dimensional unit simplex, called the {\em elementary hypercycle}. 
$ \dot{x}_i = k_i x_ix_{(i-1) \bmod n} - x_i \sum_{j} k_j x_jx_{j-1}$
for $i \in \{ 1,\ldots,n\}.$ 
 Here $k_i>0$ for all $i.$  
This system has a fixpoint (easy to compute) but for $n \geq 5$ the fixpoint is unstable. Thus, restricting their attention to the simplex minus a small ball around this unstable fixpoint, Eigen and Schuster conjectured that for $n \geq 5$ there is always a stable limit cycle which lies strictly in the interior of the simplex. Remarkably, the existence of such a limit cycle was at the core of their arguments explaining the origins of life from the proverbial primordial soup. The existence of such a cycle was proved (using the two-dimensional Poincar\'e-Bendixson theorem)  \cite{HofbauerMS91}, Theorem \ref{thm:eps-complexity} implies that we can compute an $\eps$-cycle in {\bf PSPACE}. Can the limit cycle of this particular system  --- that is to say, Life! --- be approached in polynomial time?

\paragraph{\bf Acknowledgments.} Many thanks to Eric Allender, Paul Goldberg  and  Ankit Gupta for helpful discussions. In particular, Paul showed us an unpublished proof of his, which enabled us to improve our complexity lower bound from {\bf PP} to {\bf PSPACE}.

\bibliographystyle{plain}
\bibliography{Poincare}

\appendix
   
 \section{The model of \cite{Ko} and the Poincar\'e-Bendixson Impossibility Theorem}\label{sec:cts}
We prove that both the decision and the search problems (checking and finding points on the limit cycle) of the continuous version of the Poincar\'e-Bendixson theorem are arbitrarily hard even when the function $f$ is polynomial time computable. Since the problems we consider involve real numbers while standard complexity classes are defined with respect to strings, the right framework to study these problems  is that of the complexity of computation on real numbers \cite{Ko,KoF82}. In this setting,  the dynamical system $\dot x = f(x)$  is presented by an oracle Turing machine computing a Lipschitz continuous and differentiable function  $f(x)$. We first review this model before stating and proving our results.

\paragraph{\bf Computational complexity of real functions}
 In this model a real number is (possibly non-uniquely) represented as a sequence of {\em dyadic rational numbers} of the form
$$ s,a_k,\ldots a_0, a_{-1}, \ldots, a_{-m},$$ 
where $s \in \{+,-\}$ and $a_i \in \{0,1\}.$ $a_k=1$ unless $k=0.$  
Let $\mathbb{D}_m$ denote these set of strings.  Each string in $\mathbb{D}_m$ is, via an abuse of notation, associated to a rational number $s \sum_{i=-m}^{k} a_i2^i$ which is a multiple of $2^{-m}.$ A real number $x$ is represented by a function $\varphi $ if $| \varphi(0^m) -x | \leq 2^{-m}.$ (Here $0^m$ denotes the string of $0$s of length $m.$)
An oracle Turing machine computes a function $f:[0,1]\mapsto \mathbb{R}$ if, given any representation of $x \in [0,1]$ as an oracle, it computes some representation of $f(x).$ Such a machine is said to run in polynomial time (space) if for any $n,$ given an oracle access to a representation of $x,$ it outputs a number which is  within $2^{-n}$ of $f(x)$ and runs in time (space) $p(n)$ for some fixed polynomial $p(\cdot).$ Thus, the machine can never access $x$ to an accuracy more than $2^{-p(n)}.$ Thus, all computable functions in this model are continuous and all polynomial-space computable functions have polynomial modulus of continuity. This definition can be straightforwardly generalized to the case when $f$ has multiple inputs and multiple outputs: $f:[0,1]^k \mapsto \mathbb{R}^l.$ 
We now define what it means for a function to have arbitrarily high complexity and state the main result of this section.

\begin{definition}
A computational problem that takes as input an oracle Turing machine that computes a function from $[0,1]^k\mapsto \mathbb{R}^\ell$ and outputs a real number is said to have {\em arbitrarily high complexity} if for every function $K:\mathbb{Z}_+\mapsto \mathbb{Z}_+$ and every Turing machine $M$ there is an input oracle Turing machine for which $M$ cannot compute the output with precision $2^{-n}$ in fewer than $K(n)$ steps.   Similarly, a computational problem that takes as input an oracle Turing machine that computes a function from $[0,1]^k\mapsto \mathbb{R}^\ell$ and output a binary (``yes'' -- ``no'') is said to have {\em arbitrarily high complexity} if for every integer $K$ and every Turing machine $M$ there is an input oracle Turing machine for which $M$ cannot compute the correct answer in fewer than $K$ steps.
\end{definition}

\begin{theorem}[Poincar\'e-Bendixson Impossibility]\label{thm:cts-uncomputable}
The following two problems have arbitrarily high complexity:
\begin{enumerate}
\item Given an oracle access to a  polynomial time computable function $f$ from a two-dimensional compact domain $T$ to itself, which is known to be Lipschitz continuous and continuously differentiable and have no fixpoints, find a point guaranteed to be on a limit cycle of the dynamical system $\dot x = f(x)$.
\item Given an oracle access to polynomial time computable function $f$ from a two-dimensional compact domain $T$ to itself, which is known to be  Lipschitz continuous and continuously differentiable and have no fixpoints, and a point $x\in T$, determine if $x$ on (or $\epsilon$-close to for some fixed $\epsilon=\nfrac{1}{4}$) 
 a limit cycle of the dynamical system $\dot x = f(x)$.
\end{enumerate}
\end{theorem}
\noindent
 Both results follow from the well known fact that polynomially computable functions can have uncomputable roots\cite{Ko}; it is also easy to see that the Lipschitzness and continuous differentiability requirements do not affect its validity:
\begin{theorem}\label{thm:zero}
The following problems have arbitrarily high complexity:  given a continuous function $\phi:[0,1]\mapsto[0,1]$, find a root $x\in [0,1]$ or determine that none exists.  This remains true even if $\phi$ is monotone, continuously differentiable and Lipschitz continuous, and even if $\phi(0)>0$ and and $\phi(1)<0$ (and hence the root is unique).
\end{theorem} 

\subsection{Proof of Theorem \ref{thm:cts-uncomputable}}
First we let the compact domain $T$ be the two-dimensional annulus with inner boundary of radius $1$ and outer boundary of radius $2.$ Further, we define our function $f$ in polar coordinates: $r, \theta.$ The radius $r \in [1,2]$ and $\theta \in [0,2\pi].$ The function $f(r,\theta)$ has two components: $(f_1(r,\theta),f_2(r,\theta))$ and the dynamical equations are $\dot{r}=f_1(r,\theta)$ and $\dot{\theta}=f_2(r,\theta).$ To prove our result, we let $f_1(r,\theta)=\phi(r-1)$ and $f_2(r,\theta)=1$ where $\phi$ is the function  as in Theorem \ref{thm:zero}. Thus, $f$ satisfies the conditions of the Poincar\'e-Bendixson theorem: $f$ is continuously differentiable, Lipschitz  and has no fixed point in $T.$ Further, it is obvious that there is exactly one limit cycle for this system which is the  circle concentric to the boundaries that passes through the unique root $\zeta$ of $\phi.$ 
Given an oracle for $f$, assume we can output two representations $(\tilde{r},\tilde{\theta})$ for a point on the limit cycle. Hence, $\tilde{r}-1$ must be a representation for $\zeta$. However, Theorem \ref{thm:zero} implies that $\phi^{-1}(0)$ has arbitrarily high time complexity in this model. Hence, finding a limit cycle for $f$ must also have arbitrarily high time complexity.  

The decision version of the theorem with arbitrarily small $\delta$ also follows straightforwardly by a simple binary search argument.  With a more complicated construction, the decision problem can be shown to be similarly intractable even if $\delta \geq 1/4$.

\end{document}